\documentclass[12pt]{amsart}
\newcommand{\Tr}{{\rm Tr\,}}

\newcommand{\R}{\mathbb{R}}
\newcommand{\E}{\mathbb{E}}

\newcommand{\sgn}{{\rm sgn\,}}

\newcommand{\eins}{\leavevmode\hbox{\small1\kern-3.8pt\normalsize1}}

\newcommand{\be}{\begin{equation}}
\newcommand{\ee}{\end{equation}}
\newcommand{\bee}{\begin{eqnarray}}
\newcommand{\eee}{\end{eqnarray}}

\newtheorem{thm}{Theorem}[section]

\newtheorem{lem}[thm]{Lemma}
\newtheorem{prop}[thm]{Proposition}
\newtheorem*{prob*}{Problem}
\newtheorem*{thm*}{Theorem}

\theoremstyle{definition}
\newtheorem{defn}[thm]{Definition}
\newtheorem{example}[thm]{Example}
\newtheorem*{defn*}{Definition}
\newtheorem{rem}[thm]{Remark}

\newtheorem*{rem*}{Remark}
\numberwithin{equation}{section}

\begin{document}
\title[Averages  of characteristic polynomials in polynomial ensembles]
{\bf{Averages of  products and ratios of characteristic polynomials in polynomial ensembles}}
\author{Gernot Akemann}
\address{Faculty of Physics, Bielefeld University, PO-Box 100131, 33501 Bielefeld, Germany,
 and School of Mathematics and Statistics,
The University of Melbourne, Parkville, VIC 3010, Australia}\email{akemann@physik.uni-bielefeld.de}
\author{Eugene Strahov}
\address{Department of Mathematics, The Hebrew University of Jerusalem, Givat Ram, Jerusalem 91904, Israel}
\email{strahov@math.huji.ac.il}
\author{Tim R. W\"urfel}
\address{Faculty of Physics, Bielefeld University, PO-Box 100131, 33501 Bielefeld, Germany}
\email{twuerfel@physik.uni-bielefeld.de}
\commby{}
\begin{abstract}
Polynomial ensembles are a sub-class of probability measures within determinantal point processes. Examples include products of independent random matrices, with applications to Lyapunov exponents, and random matrices with an external field, that may serve as schematic models of quantum field theories with temperature.
We first analyse expectation values of ratios of an equal number of characteristic polynomials in general polynomial ensembles. Using Schur polynomials we show that  polynomial ensembles constitute Giambelli compatible point processes, leading to a determinant formula for such ratios
as in classical ensembles of random matrices.
In the second part we introduce invertible polynomial ensembles given e.g. by  random matrices with an external field. Expectation values of arbitrary ratios of characteristic polynomials are expressed in terms of multiple contour integrals.
This generalises previous findings by Fyodorov, Grela, and Strahov 
for a single ratio
in the context of eigenvector statistics in the complex Ginibre ensemble.
\end{abstract}

\maketitle
\section{Introduction}
\label{S:1}

In this paper we study correlation functions of characteristic polynomials in a sub-class of determinantal random point processes.
They are called polynomial ensembles \cite{Arno} and belong to biorthogonal ensembles in the sense of Borodin \cite{Borodin}.
Polynomial ensembles are characterised by the fact that
one of the two determinants in the joint density of points is given by a Vandermonde determinant, while the other one is kept general.
Thus they are generalising the classical ensembles of Gaussian random matrices \cite{Mehta}.
Polynomial ensembles appear in various contexts as the joint distribution of eigenvalues (or singular values) of random matrices, see \cite{Guhr1,BrezinHikami1,DesFor,DesFor2,AKW}. They enjoy many invariance properties
on the level of joint density, kernel and bi-orthogonal functions \cite{ArnoPE,KK2}, and provide examples for realisations of multiple orthogonal polynomials, see e.g. \cite{Arno,DesFor2,BleherK1} and Muttalib-Borodin ensembles \cite{Muttalib, Borodin}.

Random matrices enjoy many different applications in physics and beyond, see \cite{Handbook} and references therein. Polynomial ensembles in particular are relevant in the following contexts: Ensembles with an external field have been introduced as a tool to count intersection numbers of moduli spaces on Riemann surfaces \cite{BHintersect}.
In the application to the quantum field theory of the strong interactions, quantum chromodynamics (QCD), they have been used as a schematic model to study the influence of temperature in the chiral phase transition \cite{PhaseDiag}. Detailed
computations of Dirac operator eigenvalues \cite{GuhrWettig,SWG} within this class of models have been restricted to supersymmetric techniques so far, that can now be addressed in the framework of biorthogonal ensembles.

Recently sums and products of random matrices have been shown to lead to  polynomial ensembles \cite{AKW, CKW, KKS} - see \cite{AI} for a review. This has important consequences for the spectrum of Lyapunov exponents, relating this multiplicative process to the additive process of Dyson's Brownian motion \cite{ABK}.
Last but not least polynomial ensembles of P\'olya type have led to a deeper understanding of the relation between singular values and complex eigenvalues, \cite{KK1,KK2} where a bijection between the respective point processes was constructed.

In this paper we  consider expectation values of products and ratios of characteristic polynomials within the class of polynomial ensembles.  While these can be used to generate multi-point resolvents and thus arbitrary $k$-point density correlation functions, as well as the kernel of bi-orthogonal polynomials,  they are of interest in their own right as well.
Examples for applications are the partition function of QCD  with an arbitrary number of fermionic flavours \cite{ShuryakVerbaar}.
In mathematics the Montgomery conjecture in conjunction with moments of the Riemann zeta-functions has lead to important insights \cite{KeatingSnaith},
where moments and correlations of characteristic polynomials relevant for more general $L$-functions were computed.

Mathematical properties of ratios of characteristic polynomials have equally received attention, and we will not be able to give full justice to the existing literature.
Based on earlier works such as \cite{Strahov} and \cite{BDS},
the determinantal structure of the expectation value of ratios of characteristic polynomials in orthogonal polynomial ensembles was expressed in several  equivalent forms, given in terms of orthogonal polynomials, their Cauchy transform or their respective kernels. This structure was generalised for products of characteristic polynomials in \cite{AV}, as well as to all symmetry classes \cite{KG}.
The universality of such ratios has been studied in several works \cite{SF,BreuerStrahov} and in particular its relation to the sine- and Airy-kernel \cite{BorodinStrahov}. New critical behaviours have been found from such ensembles as well \cite{BrezinHikami2} and 
their universality was discussed in  \cite{BleherK}.

Moving to polynomial ensembles, expectation values of products are easy to evaluate by including them into the Vandermonde determinant, just as for orthogonal polynomial ensembles. Determinantal formulas for expectation values of characteristic polynomials and their inverse have been derived, see e.g. \cite{DesFor, ForresterLiu, BleherK1}. A duality in the number of products and matrix dimension, which is well known for the classical ensembles,  holds also in this external field model \cite{BH08}. The kernel for general polynomial ensembles has been expressed in terms of the residue of a single ratio of characteristic polynomials in \cite{DesFor}, see also  \cite{BorodinStrahov,GGK}.
Most recently the study of eigenvector statistics of random matrices has seen a revival, and also in this context  expectation values of ratios of characteristic polynomials in polynomial ensembles arise \cite{YanVec,Fyodorov}. This has been one of the starting points of the present work.

The outline of the paper is as follows. In Section \ref{S:2} we introduce polynomial ensembles, provide several examples, and state 
the main results of the present paper. In particular, Theorem \ref{TheoremPolynomialIsGiambeliCompatible} says that \textit{any polynomial ensemble is 
a Giambelli compatible point process} in the sense of Borodin, Olshanski, and Strahov \cite{Giambelli}.   
This leads to Theorem \ref{GCThm}, expressing the expectation value of the ratio of an equal number of characteristic polynomials as a determinant of a single ratio, generalising \cite[Theorem 3.3]{BDS} to polynomial ensembles.
In Section \ref{S:4} we introduce
a more restricted class of polynomial ensembles which we call \textit{invertible}. Here, we give a nested multiple complex contour integral representation
 for  general ratios of characteristic polynomials in Theorem \ref{MainTheorem}. The number of integrals only depends on the number of characteristic polynomials, but not on the number of points $N$ of the point process.  This generalises the results of
\cite[Theorem 5.1]{Fyodorov} to  rectangular random matrices, in the presence of an arbitrary number of characteristic polynomials.
Several examples are given that belong to the class of invertible polynomial ensembles, including the external field models.
Sections \ref{S:3} 
 and \ref{Sec:4}
 are devoted to the proofs of the results stated in Section  \ref{S:2}.
Section \ref{SectionSpecialCases} contains some special cases, and comparison with the work by Fyodorov, Grela, and Strahov \cite{Fyodorov}.
Finally, Appendix A collects properties of the Vandermonde determinant, when adding or removing factors.

\section{Definitions and statement of results}
\label{S:2}
\subsection{Polynomial ensembles}\label{SectionPolynomialEnsembles}
We introduce  polynomial ensembles following \cite{Arno}.
They are defined by the probability density function   on $I^n$,  where $I \subseteq \mathbb{R}$ is an interval.
The  probability density function is given by
\begin{equation} \label{densitygeneral}
\mathcal{P}(x_1,\ldots,x_N) = \frac{1}{\mathcal{Z}_N} \Delta_N(x_1,\ldots,x_N)
\det [\varphi_l(x_k)]_{k,l=1}^{N}\ ,
\end{equation}
where $\Delta_N(x_1,\ldots,x_N) = \prod_{1\leq i < j \leq N} (x_i-x_j)=\det\left[x_i^{N-j}\right]_{i,j=1}^N$ is the Vandermonde determinant of $N$ variables.
The  $\varphi_{1},\ldots,\varphi_{_N}$ are certain integrable real-valued functions on $I$, such that the normalisation constant $\mathcal{Z}_N$
\begin{equation}\label{partitionfunction}
\mathcal{Z}_N = \left( \prod_{n=1}^{N} \int_I dx_n \right)  \Delta_N(x_1,\ldots,x_N) \det [\varphi_l(x_k)]_{k,l=1}^{N}= N! (-1)^{N(N-1)/2} \det[G] \ ,
\end{equation}
exists and is non-zero. The constant $\mathcal{Z}_N$ is also called partition function in the physics literature. Polynomial ensembles  are formed by eigenvalues (or singular values) of certain  $N\times N$ random matrices $H$, see examples below.
Here, the matrix $G=(g_{k,l})_{k,l=1}^{N}$ is the invertible generalised moment matrix with entries
\be\label{gentries}
g_{k,l} = \int_I dx\, x^{k-1} \varphi_{l}(x)\ .
\ee
The second equality in \eqref{partitionfunction} follows using \eqref{Vandermonde} and the Andr\'ei\'ef integral formula,
\be
\label{Andreief}
\left( \prod_{n=1}^{N} \int_I dx_n \right)  \det[\psi_l(x_k)]_{k,l=1}^{N} \det [\phi_l(x_k)]_{k,l=1}^{N}
=N!
\det\left[\int_Idx \psi_k(x)\phi_l(x)\right]_{k,l=1}^{N} \ ,
\ee
valid for any two sets of integrable functions $\psi_k$ and $\phi_l$.
We will now give some explicit realisations of polynomial ensembles in terms of random matrices.
The simplest example for a polynomial ensembles is given by the eigenvalues of $N\times N$ complex Hermitian random matrices $H$ from the Gaussian Unitary Ensembles (GUE), defined by the probability measure
\be
\label{GUE}
{P}_{\rm GUE}(H)dH= c_N \exp[-\Tr[H^2]]dH\ ,\quad c_N= {2^{\frac{N(N-1)}{2}}}{\pi^{-\frac{N^2}{2}}}\ .
\ee
The  probability density function of the real eigenvalues $x_1,\ldots,x_N$ of $H$ reads \cite{Mehta}
\be
\label{GUEev}
\mathcal{P}_{\rm GUE}(x_1,\ldots,x_N) = \frac{1}{\mathcal{Z}_N^{\rm GUE}} \Delta_N(x_1,\ldots,x_N)^2 \exp\left[-\sum_{j=1}^Nx_j^2\right].
\ee
This is a polynomial ensemble where the resulting $\varphi$-functions, $\varphi_k(x)=x^{N-k}e^{-x^2}$, are obtained
after multiplying the exponential factors into one of the Vandermonde determinants.
Note that the GUE is an \textit{orthogonal} polynomial ensemble.

The GUE with an external source or field \cite{Guhr1,BrezinHikami1} contains an additional constant, deterministic Hermitian matrix $A$ of size $N\times N$ that we choose to be diagonal here, $A=\mbox{diag}(a_1,\ldots,a_N)$ with $a_j\in\mathbb{R}$ for $j=1,\ldots,N$, without loss of generality. It will constitute our first main example and is defined by the probability measure
\be
\label{GUEext}
{P}_{\rm ext1}(H)dH= {c}_N \exp[-\Tr[ (H-A)^2]]dH\ ,
\ee
with the probability density function
\be
\label{ext1ev}
\mathcal{P}_{\rm ext1}(x_1,\ldots,x_N) = \frac{1}{\mathcal{Z}_N^{\rm ext1}} \Delta_N(x_1,\ldots,x_N) \det\left[\exp[-(x_j-a_k)^2]\right]_{j,k=1}^N.
\ee
The resulting  $\varphi_k(x)=e^{-(x-a_k)^2}$ follows from the Harish-Chandra--Itzykson--Zuber integral \cite{HarishChandra,ItzyksonZuber}  and from multiplying the Gaussian term inside the determinant. We refer to \cite{BrezinHikami1} for the derivation.
Notice that the second determinant in \eqref{ext1ev} cannot be reduced to a Vandermonde determinant in general.

Our second main example is the chiral GUE with an external source, cf. \cite{DesFor}.
It is defined in terms of a  complex non-Hermitian $N\times (N+\nu)$ dimensional random matrix $X$ and a deterministic matrix $A$ of equal size, with $\nu\geq0$. Again without loss of generality we can choose $AA^\dag=\mbox{diag}(a_1,\ldots,a_N)$, with elements $a_j\in\mathbb{R}_+$ for $j=1,\ldots,N$. The ensemble is defined by
\be
P_{\rm ext2}(X)dX  = \hat{c}_N \exp\left[-\Tr[ (X-A)(X^\dag -A^\dag)\right]dX ,\quad \hat{c}_N=c_N \pi^{-N(N+\nu)}.
\label{chGUEext}
\ee
At vanishing $A$ it reduces to the chiral GUE also called 
complex Wishart or Laguerre unitary ensemble. The probability density function of the real positive eigenvalues $x_1,\ldots,x_N$ of $XX^\dag$ reads,
\be
\label{ext2ev}
\mathcal{P}_{\rm ext2}(x_1,\ldots,x_N) = \frac{1}{\mathcal{Z}_N^{\rm ext2}} \Delta_N(x_1,\ldots,x_N)
\det\left[x_j^{\nu/2}e^{-(x_j+a_k)}I_\nu\left(2\sqrt{a_kx_j}\right)\right]_{j,k=1}^N.
\ee
The modified Bessel function of second kind $I_\nu$ inside
$\varphi_k(x)=x^{\nu/2}e^{-(x+a_k)}I_\nu(2\sqrt{a_kx})$
follows from the
Berezin-Karpelevich integral formula, cf. \cite{SchlittgenWettig}. In principle we may also allow the parameter $\nu>-1$ to take real values.

In the application to QCD at finite temperature typically the density \eqref{chGUEext} is endowed with $N_f$ extra terms, $P(X) \to P(X)\prod_{f=1}^{N_f}\det[XX^\dag+m_f^2\eins_N]$, with $m_{f=1,\ldots,N_f}\in\mathbb{R}$, that correspond to $N_f$ Fermion flavours with masses $m_f$, see e.g. \cite{SWG} which also motivates the present study.
We would like to mention that the expectation value of the ratio of two characteristic polynomials  studied in \cite{Fyodorov} follows from the above ensemble, when setting $\nu=0$ and letting $m_f\to0$ for all $f=1,\ldots,N_f$. This leads to a polynomial ensemble with $\varphi_k(x)=x^{\mathcal{L}}e^{-(x+a_k)}I_0(2\sqrt{a_kx})$ of \cite{Fyodorov}, with $\mathcal{L}= N_f$.

Further examples have been given already in the introduction, including the singular values of products of independent random matrices, see \cite{AI} for a review, where $\varphi_k(x)$ is given by a special function, the Mejier G-function, and more generally  Poly\'a ensembles \cite{KK1,KK2}.
Notice that when also  the Vandermonde determinant in \eqref{densitygeneral} is replaced be a general determinant, as in the Andr\'ei\'ef integration formula \eqref{Andreief}, we are back to biorthogonal ensembles \cite{Borodin} - an explicit example can be found in \cite{AStr}. For this class our methods below will not apply in general.

\subsection{Polynomials ensembles as Giambelli compatible point processes}
In this section we adopt notation and definitions from Macdonald \cite{MacDonald}. Let $\Lambda$ be the algebra of symmetric functions.
The Schur functions $s_{\lambda}$ indexed by Young diagrams $\lambda$ form an orthonormal basis in $\Lambda$. Recall that Young diagrams can be written in the Frobenius notation, namely
$$
\lambda=\left(p_1,\ldots,p_d|q_1,\ldots,q_d\right),
$$
where $d$ equals the number of boxes on the diagonal of $\lambda$,
$p_j$ with $j=1,\ldots,d$ denotes the number of boxes
in the $j$-th row of $\lambda$
to the right of the diagonal,
and $q_l$ with $l=1,\ldots,d$  denotes the number of boxes in the $l$-th column of $\lambda$ below the diagonal.
The Schur functions satisfy the \textit{Giambelli formula}:
\begin{equation}
s_{\left(p_1,\ldots,p_d|q_1,\ldots,q_d\right)}=\det\left[s_{(p_i|q_j)}\right]_{i,j=1}^d.
\end{equation}
The Schur polynomial $s_{\lambda}\left(x_1,\ldots,x_N\right)$ is the specialization of $s_{\lambda}$ to the variables
$x_1$, $\ldots$, $x_N$. The Schur polynomial  $s_{\lambda}\left(x_1,\ldots,x_N\right)$ corresponding to the Young diagram $\lambda$
with $l(\lambda)\leq N$ rows of lengths $\lambda_1\geq ... \geq \lambda_{l(\lambda)}>0$, can be defined by
\be\label{Schurfunction}
s_\lambda(x_1,\ldots,x_N)=\frac{1}{\Delta_N(x_1,\ldots,x_N)} \det \left[ x_i^{\lambda_j + N -j} \right]_{i,j=1}^{N}.
\ee
If $l(\lambda)>N$, then $s_\lambda(x_1,\ldots,x_N) \equiv 0$  (by definition).

The \textit{Giambelli compatible} point processes form a class of point processes whose different probabilistic quantities
of interest can be studied using the Schur symmetric functions. This class of point processes was introduced in
Borodin, Olshanski, and Strahov \cite{Giambelli} to prove determinantal identities for averages of analogs of characteristic polynomials
for ensembles originating from Random Matrix Theory, the theory of random partitions, and from representation theory of the infinite symmetric group.
In the context of random point processes formed by $N$-point random configurations on a subset of $\mathbb{R}$, the Giambelli compatible point processes can be defined as
follows.
\begin{defn}\label{Giambellidef}
Assume that a point process is formed by an $N$-point configuration\newline $\left(x_1,\ldots,x_N\right)$ on $I\subseteq\mathbb{R}$.
If the Giambelli formula
\be\label{Giambelliformula}
s_{(p_1,\ldots,p_d\vert q_1,\ldots,q_d)}(x_1,\ldots,x_N) = \det \left[ s_{(p_i \vert q_j)}(x_1,\ldots,x_N) \right]_{i,j=1}^{d}
\ee
(valid for the Schur polynomial  $s_\lambda(x_1,\ldots,x_N)$  parameterized by an arbitrary Young diagram $\lambda\left(p_1,\ldots,p_d|q_1,\ldots,q_d\right)$) can be extended to averages, i.e.
\be\label{Giambellicompatibleformula}
\mathbb{E}\left[
s_{(p_1,\ldots,p_d\vert q_1,\ldots,q_d)}(x_1,\ldots,x_N)
\right] = \det \left[ \mathbb{E}\left[ s_{(p_i \vert q_j)}(x_1,\ldots,x_N) \right] \right]_{i,j=1}^{d}\ ,
\ee
then the random point process is called Giambelli compatible point process.
\end{defn}
In the present paper we show that the polynomial ensembles introduced in Section \ref{SectionPolynomialEnsembles}
can be understood as Giambelli compatible point processes. Namely, the following Theorem holds true.
\begin{thm}\label{TheoremPolynomialIsGiambeliCompatible}
Any polynomial ensemble in the sense of Section \ref{SectionPolynomialEnsembles} is a Giambelli compatible point process.
\end{thm}
As it is explained in Borodin, Olshanski, and Strahov \cite{Giambelli} the Giambelli compatibility of point processes
implies determinantal formulas for averages of ratios of characteristic polynomials. Namely, we obtain
\begin{thm}\label{GCThm}
Assume that $x_1,\ldots,x_N$ form a polynomial ensemble.
Let $u_1,\ldots,u_M \in \mathbb{C} \backslash \mathbb{R}$ and $z_1,\ldots,z_M \in \mathbb{C}$ for any $M \in \mathbb{N}$ be pairwise distinct variables. Then
\be
\mathbb{E}\left[ \prod_{m=1}^{M} \frac{ D_N(z_m)}{ D_N(u_m)} \right] = \left[ \det \left( \frac{1}{u_i-z_j} \right)_{i,j=1}^{M} \right]^{-1} \det \left[ \frac{1}{u_i - z_j} \mathbb{E} \left( \frac{D_N(z_j)}{D_N(u_i)} \right) \right]_{i,j=1}^{M},
\ee
where $D_N(z)=\prod_{n=1}^N(z-x_n)$ denotes the characteristic polynomial associated with the random variables $x_1$, $\ldots$, $x_N$.
\end{thm}

\subsection{Averages of arbitrary ratios of characteristic polynomials in invertible ensembles}
\label{S:4}
In this section we present our results for  arbitrary ratios of characteristic polynomials,
\be
\mathbb{E}\left[\frac{\prod_{m=1}^{M} D_{N}(z_m)}{\prod_{l=1}^{L} D_{N}(y_l) } \right],
\label{Estart}
\ee
allowing the number of characteristic polynomials in the numerator $M$ and denominator, $L\leq N$, to differ. As before we will assume the parameters
 $y_1,\ldots,y_L \in \mathbb{C} \backslash \mathbb{R}$ and $z_1,\ldots,z_M \in \mathbb{C}$ to be pairwise distinct. We will not consider the most general polynomial ensembles \eqref{densitygeneral} here, but consider functions $\varphi_j(x)$ that satisfy certain conditions to be specified below.
\begin{defn}\label{invert-def}
Consider a polynomial ensemble defined by the probability density function (\ref{densitygeneral}).
Assume that
$\varphi_l(x)=\varphi(a_l,x)$ for $l=1,\ldots,N$, (where $a_1,\ldots,a_N$ are real parameters)
is anaytic in both arguments, 
and that there exists
a family $\left\{\pi_k\right\}_{k=0}^{\infty}$ of monic polynomials such that each 
polynomial 
$\pi_k$ 
of degree $k$
can be represented as
\begin{equation}\label{Gcondition}
\pi_{k}(a)=\int_I dx x^{k}\varphi(a,x), \quad k=0,1,\ldots.
\end{equation}
In addition, assume that equation (\ref{Gcondition}) is invertible, i.e. there exists a function $F: I'\times\mathbb{C}\rightarrow\mathbb{C}$ such that
\begin{equation}\label{Picondition}
z^{k}=\int_{I^\prime}ds F(s,z)\pi_{k}(s),\quad k=0,1,\ldots,
\end{equation}
where $I^\prime$ is a certain contour in the complex plane. Then, we will refer to such a polynomial ensemble as an \textit{invertible} ensemble.
\end{defn}
\begin{rem}
Condition \eqref{Gcondition} together with \eqref{partitionfunction}  immediately implies that for invertible polynomial ensembles the normalising partition function
simplifies as follows:
\be
\label{ZNinv}
\mathcal{Z}_N=N!\Delta_N(a_1,\ldots,a_N).
\ee
Here, we use that in \eqref{Vandermonde} the determinant of monomials equals that of arbitrary monic polynomials.
\end{rem}
We will now present two examples for polynomial ensembles of invertible type according to Definition \ref{invert-def}, before commenting on the general class of such ensembles.

\begin{example}\label{ex1}
Our first example is given by the GUE with external field \eqref{ext1ev}. Here, the eigenvalues take real values, $I=\mathbb{R}$, and the functions $\varphi_l(x)$ can be chosen
as
\be
\varphi_l(x)=\varphi(a_l,x)=\frac{e^{-(x-a_l)^2}}{\sqrt{\pi}},
\label{phiHerm}
\ee
which are analytic.
From \cite[8.951]{Gradshteyn} we know the following representation of the standard Hermite polynomials $H_n(t)$ of degree $n$,
\begin{equation}\label{repHermite}
H_n(t) = \frac{(2i)^n}{\sqrt{\pi}} \int_{-\infty}^{\infty} dx e^{-(x+it)^2} x^n\ ,
\end{equation}
that can be  made monic as follows,
$2^{-n} H_n(x)=x^n+O(x^{n-2})$. This leads to the integral
\be
(2i)^{-n}H_n(ia)=\frac{1}{\sqrt{\pi}}\int_{-\infty}^\infty ds s^n e^{-(s-a)^2}\ ,
\ee
from which we can read off
\be\label{GmatrixentriesHermite}
\pi_{k}(a)=\int_{-\infty}^{\infty} dx x^{k}\frac{e^{-(x-a)^2}}{\sqrt{\pi}}\ ,
\ee
with $\pi_{k}(a)=(2i)^{-k}H_{k}(ia)$, for $k=0,1,\ldots$, which is again monic.
Thus condition \eqref{Gcondition} is satisfied.

For the second condition \eqref{Picondition} we use the integral \cite[7.374.6]{Gradshteyn}
\begin{equation}\label{monomialHermite}
y^{n} = \frac{1}{\sqrt{\pi}} \int_{-\infty}^{\infty} dx \ 2^{-n} H_n(x) e^{-(x-y)^2}
\ .
\end{equation}
Renaming $y=iz$ and $x=is$ we obtain
\be
z^{k}=\int_{I^\prime} ds F(s,z)\pi_{k}(s)\ , \quad \mbox{for}\  k=0,1,\ldots
\ee
with  $I^\prime=i\mathbb{R}$ and $F(s,z)= \frac{i}{\sqrt{\pi}}e^{(s-z)^2}$.
\end{example}
\begin{rem}
Example \ref{ex1} is the simplest case of a much wider class of polynomial ensembles of P\'olya type 
convolved with fixed matrices, as introduced in \cite[Theorem II.3]{MK}. Such polynomials ensembles are generalising the form \eqref{phiHerm} to
\be
\varphi(a_l,x) =f(x-a_l)\ ,
\ee
 such that $f$ is $(N-1)$-times differentiable on $\mathbb{R}$,  
analytic on $\mathbb{C}$, 
and the moments of its derivatives exist\footnote{The unconvoluted polynomial ensemble has $\varphi_j(x)=\partial^j f(x)/\partial x^j$.},
\be
\left|\int_{-\infty}^\infty dx x^{k}\frac{\partial^j f(x)}{\partial x^j}\right|<\infty\ ,\quad \forall k,j=0,1,\ldots, N-1.
\ee
It immediately follows that its generalised moment matrix leads to polynomials, upon shifting the integration variable, and thus \eqref{Gcondition} is satisfied. It is not too difficult to show using Fourier transformation of $f$ that also condition \eqref{Picondition} of Definition \ref{invert-def} is satisfied and thus these ensembles are invertible.
\end{rem}

\begin{example}\label{ex2}
Our second example is the chiral GUE with external field \eqref{ext2ev} having   $I = \mathbb{R}_+$ and functions $\varphi_l(x)$ can be chosen as
\begin{equation}\label{varphiLaguerre}
\varphi_l(x)=\varphi(a_l, x) = \left(\frac{x}{a_l}\right)^{\nu/2} e^{-\left(x+a_l\right)}I_\nu (2\sqrt{a_l x})\ ,
\end{equation}
which is analytic, 
with the $a_l$ 
positive real numbers.
The following integral is known, see e.g. \cite[6.631.10]{Gradshteyn} after analytic continuation,
\begin{equation}\label{gintegral}
\int_0^{\infty} x^{n+\frac{\nu}{2}} e^{-x} I_\nu ( 2 \sqrt{ax}) dx = n! a^{\nu/2} e^{a} L_n^{\nu} \left( -a \right)\ .
\end{equation}
Here, $L_n^\nu(y)$ is the standard generalised Laguerre polynomial of degree $n$, which is made monic as follows, $n!L_n^\nu(-x)=x^n+O(x^{n-1})$.
Then, the first condition \eqref{Gcondition} is satisfied,
\be\label{GmatrixentriesLaguerre}
\pi_{k}(a)=\int_0^{\infty} dx x^{k} \left(\frac{x}{a}\right)^{\frac{\nu}{2}} e^{-(x+a)} I_\nu (2\sqrt{ax}),
\ee
with  $\pi_{k}(a)=k!L_{k}^\nu(-a)$ for $k=0,1,\ldots.$

For the second condition \eqref{Picondition} we consider the following integral, see \cite[7.421.6]{Gradshteyn}, which is also called Hankel transform,
\be
\int_0^\infty dt t^{\nu/2}e^{-t} n!L_n^\nu(t)
J_\nu\left(2\sqrt{zt}\right)=z^nz^{\nu/2}e^{-z}.
\ee
Bringing factors on the other side and making the substitution $t=-s$ to make the same monic polynomials $n!L_n^\nu(-s)$ as above appear in the integrand, we obtain after using $I_\nu(x)=i^{-\nu}J_\nu(ix)$
\be\label{zkexample2}
z^{k}=\int_{I^\prime}ds F(s,z) \pi_{k}(s)\ , \quad \mbox{for}\  k=0,1,\ldots
\ee
with  $F(s,z)= (-1)^\nu \left(\frac{s}{z}\right)^{\nu/2}e^{s+z} I_\nu\left(2\sqrt{zs} \right)$ and $I^\prime = \mathbb{R}_{-}$.
\end{example}

Now we state the second main result of the present paper which gives a formula for
averages of  products and ratios of characteristic polynomials in the case of invertible ensembles.
\begin{thm}\label{MainTheorem}Consider a polynomial ensemble (\ref{densitygeneral}) formed by $x_1$, $\ldots$, $x_N$, and assume that this ensemble is invertible in the sense of
Definition \ref{invert-def}.  Then
we have for $L\leq N$
\be\label{vevgen}
\begin{split}
\mathbb{E}\left[\frac{\prod_{m=1}^{M} D_{N}(z_m)}{\prod_{l=1}^{L} D_{N}(y_l) } \right]
=&\frac{ (-1)^{
\frac{L(L-1)}{2}} }{L!\Delta_M(z_1,\ldots,z_M)}
\left[\prod_{j=1}^M \int_{I^\prime}ds_j F(s_j,z_j)
\prod_{n=1}^N(s_j-a_{n})\!\right]\!\Delta_M(s_1,\ldots,s_M)
\\
&\times
\left[ \prod_{l=1}^{L} \int_I dv_l \left(\frac{v_{l}}{y_l}\right)^{N-L}
\frac{\prod_{m=1}^M(z_m-v_l)}{\prod_{j=1}^{L} (y_j - v_{l})}\right]  \Delta_L(v_1,\ldots,v_L)
\\
&\times
 \left[\prod_{l=1}^{L} \oint_{C_l} \frac{du_l}{2\pi i} \frac{1}{\prod_{n=1}^{N} (u_{l}-a_n)}
\frac{\varphi(u_{l},v_l)}{\prod_{j=1}^M(s_j-u_{l})}\right]
\Delta_L(u_{1},\ldots,u_{L})\ ,
\end{split}
\ee
where $D_N(z)=\prod_{n=1}^N(z-x_n)$ denotes the characteristic polynomial associated with the random variables $x_1$, $\ldots$, $x_N$,
the parameters
$y_1,\ldots,y_L \in \mathbb{C} \backslash \mathbb{R}$ and $z_1,\ldots,z_M \in \mathbb{C}$ are pairwise distinct,
and all contours $C_l$ with $l=1,\ldots,N$ encircle the points $a_1,\ldots,a_N$ counter-clockwise.
\end{thm}
 We note that Theorem \ref{MainTheorem} generalises Theorem 5.1 in \cite{Fyodorov} for the ratio of two characteristic polynomials, derived for the polynomial ensemble with $\varphi(a,x)=x^{\mathcal{L}}e^{-x}I_0(2\sqrt{ax})$, to general ratios in invertible polynomial ensembles. Clearly it is well suited for the asymptotic analysis when $N\to\infty$ as the number of integrations does not depend on $N$.
 
 \subsection{A formula for the correlation kernel for invertible ensembles}
 It is well known that each polynomial ensemble is a determinantal process.
For invertible polynomial ensembles  (see Definition \ref{invert-def})  Theorem \ref{MainTheorem}  enables us to deduce a double contour
integration formula for the correlation kernel.
\begin{prop}\label{PropositionKernel} Consider an invertible polynomial ensemble, i.e. a polynomial ensemble defined by
(\ref{densitygeneral}), where the functions $\varphi_l(x)=\varphi(a_l,x)$ satisfy the conditions specified in Definition  \ref{invert-def}. The correlation kernel
$K_N(x,y)$ of this ensemble can be written as
\begin{equation}
K_N(x,y)=\frac{1}{2\pi i}\int\limits_{I^\prime}dsF(s,x)\prod_{n=1}^N\left(s-a_n\right)\oint\limits_{C}du\frac{\varphi(u,y)}{(s-u)\prod_{n=1}^N\left(u-a_n\right)},
\end{equation}
where $C$ encircles the points $a_1$, $\ldots$, $a_N$ counter-clockwise, and where  $\varphi(u,y)$ and $F(s,x)$ are defined by equations (\ref{Gcondition})and  (\ref{Picondition})
correspondingly.
\end{prop}
\begin{proof}We use the following fact valid for any polynomial ensemble formed by $x_1$, $\ldots$, $x_N$ on $I\subseteq\R$,  see Ref. \cite{DesFor}\footnote{Because we take the residue of the right hand side at $z=y$, any ratio $f(v)/f(z)$ can be multiplied under the integral for regular functions $f$, without changing the value of the kernel.}.
Assume that
\begin{equation}
\mathbb{E}\left(\prod\limits_{k=1}^N\frac{x-x_k}{z-x_k}\right)=\int\limits_Idv\frac{x-v}{z-v}\Phi_N(x,v),
\end{equation}
where the function $v\rightarrow\Phi_N(x,v)$ is analytic at $y$, $y\in I$. Then the correlation kernel of the determinantal process formed by
$x_1,\ldots,x_N$ is  given by
$$
K_N(x,y)=\Phi_N(x,y).
$$
In our case Theorem \ref{MainTheorem} gives
\begin{equation}
\begin{split}
&\mathbb{E}\left(\prod\limits_{k=1}^N\frac{x-x_k}{z-x_k}\right)=
\frac{1}{2\pi i}\int\limits_{I}dv\left(\frac{v}{z}\right)^{N-1}\frac{x-v}{z-v}\\
&\times\left[\int\limits_{I^\prime}dsF(s,x)\prod_{n=1}^N\left(s-a_n\right)\oint\limits_{C}du\frac{\varphi(u,v)}{(s-u)\prod_{n=1}^N\left(u-a_n\right)}\right],
\end{split}
\end{equation}
which leads to the formula for the correlation kernel in the statement of the Proposition.
\end{proof}
\section{Proof of Theorem \ref{TheoremPolynomialIsGiambeliCompatible}}\label{S:3}
Let $x_1$, $\ldots$, $x_N$ form a polynomial ensemble on $I^N$, where $I\subseteq\R$. The probability density function of this ensemble is defined by
equation (\ref{densitygeneral}). Denote by $\widetilde{s}_{\lambda}$ the expectation of the Schur polynomial $s_{\lambda}\left(x_1,\ldots,x_N\right)$
with respect to this ensemble,
\begin{equation}\label{SchurExpectation}
\widetilde{s}_{\lambda}=\E\left(s_{\lambda}\left(x_1,\ldots,x_N\right)\right).
\end{equation}
Our aim is to show that $\widetilde{s}_{\lambda}$ satisfies the Giambelli formula, i.e.
\begin{equation}\label{GiambelliForExpectations}
\widetilde{s}_{\lambda}=\det\left[\widetilde{s}_{\left(p_i|q_j\right)}\right]_{i,j=1}^d,
\end{equation}
where $\lambda$ is an arbitrary Young diagram, $\lambda=\left(p_1,\ldots,p_d\vert q_1,\ldots,q_d\right)$
in the Frobenius coordinates. According to Definition \ref{Giambellidef} this will mean that the polynomial ensemble under considerations
is a Giambelli compatible point process.

The proof of equation (\ref{GiambelliForExpectations}) below is based on the following general fact due to Macdonald,
see Macdonald \cite{MacDonald}, Example I.3.21.
\begin{prop}\label{MacDonaldProp}
Let $\{ h_{r,s} \}$ with integer $r \in \mathbb{Z}$ and non-negative integer $s\in \mathbb{N}$ be a collection of commuting indeterminates such that  we have
\be\label{CondH}
\forall s \in \mathbb{N}: h_{0,s} = 1 \ \mbox{and} \ \forall r<0\  h_{r,s} = 0 \ ,
\ee
and set
\begin{equation}
\widetilde{s}_{\lambda}=\det \left[ h_{\lambda_i-i+j,j-1} \right]_{i,j=1}^{k},
\end{equation}
where $k$ is any number such that $k\geq l(\lambda)$. Then we have
\begin{equation}
\widetilde{s}_\lambda = \det \left[ \widetilde{s}_{(p_i \vert q_j )} \right]_{i,j=1}^{d},
\end{equation}
where $\lambda$ is an arbitrary Young diagram, $\lambda=\left(p_1,\ldots,p_d\vert q_1,\ldots,q_d\right)$
in the Frobenius coordinates.
\end{prop}
Clearly, in order to apply Proposition \ref{MacDonaldProp}  to $\widetilde{s}_{\lambda}$ defined by equation (\ref{SchurExpectation})
we need to construct a collection of indeterminates $\{ h_{r,s} \}$ such that
\begin{equation}\label{g1}
\E\left(s_{\lambda}\left(x_1,\ldots,x_N\right)\right)=\det \left[ h_{\lambda_i-i+j,j-1} \right]_{i,j=1}^{k}
\end{equation}
will hold true for an arbitrary Young diagram $\lambda$, for an arbitrary $k\geq l(\lambda)$, and such that condition (\ref{CondH})
will be satisfied.

By Andr\'ei\'ef's integration formula \eqref{Andreief} and the expression for the normalisation constant $\mathcal{Z}_N$ \eqref{partitionfunction} we can write
\be\label{averageSchur}
\begin{split}
\mathbb{E}\left[ s_\lambda(x_1,\ldots,x_N) \right] = \frac{\det \left[ \int_I dx  x^{\lambda_i +N-i}\varphi_j(x) \right]_{i,j=1}^{N}}{\det \left[ \int_I dx  x^{N-i}\varphi_j(x) \right]_{i,j=1}^{N}}\ ,
\end{split}
\ee
where we used \eqref{Vandermonde} and equation (\ref{Schurfunction}).
Notice that at this point it matters that we consider polynomial ensembles and not more general bi-orthogonal ensembles. In the latter case the Vandermonde determinant in the denominator of the Schur function \eqref{Schurfunction} would not cancel, the Andr\'ei\'ef formula would not apply and we would not know how to compute such expectation values.
Set
\begin{equation}
A_{n,m}=\int_I dx  x^{n}\varphi_m(x);\;\; n=0,1,\ldots;\;\;m=1,\ldots, N,
\end{equation}
and denote by $Q=\left(Q_{i,j}\right)_{i,j=1}^N$ the inverse\footnote{Notice that due to \eqref{gentries} we have $\det [G]=(-1)^{N(N-1)/2}\det[\tilde{G}]$.}  of $\tilde{G}=\left(\tilde{g}_{i,j}\right)_{i,j=1}^N$, where
$\tilde{g}_{i,j}=\int_Idxx^{N-i}\varphi_j(x)$. With this notation we can rewrite equation (\ref{averageSchur}) as
\begin{equation}\label{g2}
\mathbb{E}\left(s_\lambda(x_1,\ldots,x_N) \right) =\det\left[\sum\limits_{\nu=1}^NA_{\lambda_i+N-i,\nu}Q_{\nu,j}\right]_{i,j=1}^N.
\end{equation}
Since $Q$ is the inverse of $\tilde{G}$, we have
\begin{equation}
\sum\limits_{j=1}^N\tilde{g}_{i,j}Q_{j,k}=\delta_{i,k},\;\; 1\leq i,k\leq N,
\end{equation}
or
\begin{equation}\label{g3}
\sum\limits_{j=1}^NA_{N-i,j}Q_{j,k}=\delta_{i,k},\;\; 1\leq i,k\leq N.
\end{equation}
The following Proposition will imply Theorem \ref{TheoremPolynomialIsGiambeliCompatible}.
\begin{prop}
Let $\{ h_{r,s} \}$, with integer $r \in \mathbb{Z}$ and non-negative integer $s\in \mathbb{Z}_{\geq 0}$,  be a collection of indeterminates defined by
\be\label{hfunctioneq}
h_{r,s} \equiv \begin{cases}
\sum_{\nu=1}^{N} A_{N+r-s-1,\nu}Q_{\nu,s+1}\text{,} & s \in \{0,1,\ldots,N-1\} \text{,} \quad r \geq 0 \text{,} \\
\delta_{r,0} \text{,} & s \geq N \text{,} \quad r \geq 0 \text{,} \\
0 \text{,} & s\geq 0 \text{,} \quad r < 0 \text{.}
\end{cases}
\ee
The collection of indeterminates $\{ h_{r,s} \}$ satisfies condition (\ref{CondH}). Moreover,
with this collection of indeterminates $\{ h_{r,s} \}$  formula (\ref{g1})
holds true for an arbitrary Young diagram $\lambda$, and for an arbitrary $k\geq l(\lambda)$.
\end{prop}
\begin{proof}
We divide the proof of into several steps. First, the collection of indeterminates $\{ h_{r,s} \}$ defined by (\ref{hfunctioneq}) is shown to satisfy condition (\ref{CondH}). 
Next, we prove that equation (\ref{g1}) holds true for an arbitrary Young diagram $\lambda$, and for an arbitrary $k\geq l(\lambda)$.\\
\textit{Step 1.} First, we want to show that
\be\label{g4}
\det \left[ h_{\lambda_i -i+j,j-1} \right]_{i,j=1}^{k}= \det \left[ h_{\lambda_i-i+j,j-1} \right]_{i,j=1}^{l(\lambda)}\ ,
\ee
for any $k\geq l(\lambda)$.

Let $\lambda$ be an arbitrary Young diagram, and assume that $k>l(\lambda)$. Consider the diagonal entries of the $k\times k$ matrix
\be
\left( h_{\lambda_i - i +j,j-1} \right)_{i,j=1}^{k}
\nonumber
\ee
for $i=j \in \{ l(\lambda)+1,\ldots,k \}$. By definition of the $h_{r,s}$ these entries are all equal to $1$, since $\lambda_i = 0$ for $i \in \{ l(\lambda)+1,\ldots,k \}$ implying $h_{0,s}=1$ by condition (\ref{CondH}).
For $r<0$  we have $h_{r,s}=0$ (see  equation (\ref{hfunctioneq})) and the matrix $\left(h_{\lambda_i-i+j,j-1}\right)_{i,j=1}^k$
has the form
\be
\left(\begin{matrix}
\star & \ldots & \star & \vert & \star & \ldots & \ldots & \ldots & \star \\
\vdots & \ddots &  \vdots & \vert  & \vdots & \ddots & \ddots& \ddots & \vdots \\
\star & \ldots & \star & \vert & \star & \ldots & \ldots & \ldots & \star \\
-- & -- & -- & -- & -- & -- & -- & -- & -- \\
0 & \ldots & 0 & \vert & 1 & \star & \ldots & \ldots & \star \\
\vdots & & \vdots & \vert & 0 & 1 & \star  &  \ldots  & \star \\
\vdots & \ddots &  \vdots & \vert  & \vdots & \ddots & \ddots & \ddots & \vdots \\
\vdots &  &  \vdots & \vert  & \vdots &  & \ddots & \ddots & \star \\
0 & \ldots & 0 & \vert & 0 &  \ldots & \ldots & 0 & 1 \\
\end{matrix}\right),
\nonumber
\ee
where the first row from the top with zeros has the label $l(\lambda)+1$, and the first column from the left with ones has the  label $l(\lambda)+1$.
The determinant of such a block matrix reduces to the product of the determinants of the blocks, which gives relation (\ref{g4}).
\\
\\
\textit{Step 2}.
Assume now that $l(\lambda) >N$. Then it trivially holds that
\be
\mathbb{E}\left(s_\lambda(x_1,\ldots,x_N) \right) = 0\ ,
\nonumber
\ee
by the very definition of the Schur polynomials. Here,  we would like  to show that it equally holds that
\be
\det \left[ h_{\lambda_i-i+j,j-1} \right]_{i,j=1}^{l(\lambda)} = 0\ ,
\nonumber
\ee
if $l(\lambda) > N$.

We have $h_{r,s} = \delta_{r,0}$ for $s\geq N$ and $r\geq 0$. This implies that the matrix $\left(h_{\lambda_i-i+j,j-1}\right)_{i,j=1}^{l(\lambda)}$, which we can write out as
\be
\left( \begin{matrix}
h_{\lambda_1,0} & \star & \ldots & \star & \vert & h_{\lambda_1+N,N} & \ldots & \ldots & h_{\lambda_1-1+l(\lambda),l(\lambda)-1} \\
\star & h_{\lambda_2,1} & \ddots & \vdots & \vert &\vdots &  &   & \vdots \\
\vdots & \ddots &\ddots & \star & \vert  & \vdots &  &   & \vdots \\
\star & \ldots & \star & h_{\lambda_N,N-1} &\vert & h_{\lambda_N+1,N} & \ldots & \ldots  & h_{\lambda_N-N+l(\lambda),l(\lambda)-1} \\
-- & -- & -- & -- & -- & -- & -- & -- & -- \\
\star & \ldots & \ldots & \star & \vert & h_{\lambda_{N+1},N} & \ldots & \ldots  & h_{\lambda_{N+1}-N-1+l(\lambda),l(\lambda)-1} \\
\vdots &  &  & \vdots & \vert & \star & \ddots &   & \vdots \\
\vdots & &  & \vdots & \vert & \vdots & \ddots & \ddots  & \vdots \\
\star & \ldots & \ldots &  \star & \vert  & \star & \ldots & \star  & h_{\lambda_{l(\lambda)},l(\lambda)-1} \\
\end{matrix} \right)
\nonumber
\ee
has the form
\be
\left( \begin{matrix}
h_{\lambda_1,0} & \star & \ldots & \star & \vert &0 & \ldots & \ldots & 0 \\
\star & h_{\lambda_2,1} & \ddots & \vdots & \vert &\vdots &  &   & \vdots \\
\vdots & \ddots &\ddots & \star & \vert  & \vdots &  &   & \vdots \\
\star & \ldots & \star & h_{\lambda_N,N-1} &\vert & 0 & \ldots & \ldots  & 0 \\
-- & -- & -- & -- & -- & -- & -- & -- & -- \\
\star & \ldots & \ldots & \star & \vert & 0 & \ldots & \ldots  & 0 \\
\vdots &  &  & \vdots & \vert & \star & \ddots &   & \vdots \\
\vdots & &  & \vdots & \vert & \vdots & \ddots & \ddots  & \vdots \\
\star & \ldots & \ldots &  \star & \vert  & \star & \ldots & \star  &0 \\
\end{matrix} \right).
\nonumber
\ee
Thus, we can again apply the formula for determinants of block matrices to obtain
\be
\det \left[ h_{\lambda_i-i+j,j-1} \right]_{i,j=1}^{l(\lambda)} = \det \left[ h_{\lambda_i-i+j,j-1} \right]_{i,j=1}^{N} \cdot 0 = 0\ ,
\nonumber
\ee
which is true for any $l(\lambda) >N$ and therefore condition (\ref{g1}) is satisfied in this case.\\
\\
\textit{Step 3}.  Now we wish to prove that
\begin{equation}\label{g5}
\sum\limits_{\nu=1}^NA_{N-i+\lambda_i,\nu}Q_{\nu,j}=h_{\lambda_i-i+j,j-1}
\end{equation}
is valid for any Young diagram with $l(\lambda)\leq N$, and for $1\leq i,j\leq N$. Assume that $\lambda_i-i+j\geq 0$.
Then (\ref{g5}) turns into the first equation in (\ref{hfunctioneq}) with $r=\lambda_i-i+j$, $s=j-1$.
Assume that $\lambda_i-i+j<0$, then $i-\lambda_i>j$. Clearly, $i-\lambda_i\in\{1,\ldots,N\}$ in this case, and we have
\begin{equation}
\sum\limits_{\nu=1}^NA_{N-i+\lambda_i,\nu}Q_{\nu,j}=\delta_{i-\lambda_i,j}=0,
\nonumber
\end{equation}
where we have used equation (\ref{g3}). Also, if $\lambda_i-i+j<0$, and $1\leq i,j\leq N$, then $h_{\lambda_i-i+j,j-1}=0$
as it follows from equation (\ref{hfunctioneq}). We conclude that (\ref{g5}) holds true for $\lambda_i-i+j<0$ as well.

Finally, the results obtained in Step 1-Step 3 together with formula (\ref{g2}) give the desired  formula (\ref{g1}).
\end{proof}

\section{Proof of Theorem \ref{MainTheorem}}\label{Sec:4}
Denoting by $S_K$ the symmetric group of a set of $K$ variables with its elements being the permutations of these, we will utilise the following Lemma
that was proven in \cite{Strahov}.
\begin{lem}\label{Strahovlemma}
Let $L$ be an integer with $1\leq L\leq N$, and let $x_1,\ldots,x_N$ and $y_1,\ldots,y_L$ denote two sets of parameters that are pairwise distinct.  Then the following identity holds
\be
\begin{split}
&\prod_{l=1}^{L} \frac{y_l^{N-L}}{\prod_{n=1}^{N} (y_l-x_n)} =
\!\sum_{\sigma \in S_N/(S_{N-L}\times S_L)}\!\!\!\!\!\
\frac{\Delta_L(x_{\sigma(1)},\ldots,x_{\sigma(L)})\Delta_{N-L}(x_{\sigma(L+1)},\ldots,x_{\sigma(N)})\prod_{n=1}^{L}x_{\sigma(n)}^{N-L}}{\Delta_N(x_{\sigma(1)},\ldots,x_{\sigma(N)})\prod_{n,l=1}^{L} (y_l - x_{\sigma(n)})}
\end{split}
\ee
on the coset of the 
permutation group.
\end{lem}
As shown in \cite{Strahov} this follows from the Cauchy-Littlewood formula and the determinantal formula for the Schur polynomials \eqref{Schurfunction}.
We can use this identity to reduce the number of variables in the inverse characteristic polynomials from $N$ to $L$. Applied to the averages of products and ratios of characteristic polynomials we obtain
\bee
\mathbb{E} \left[\frac{\prod_{m=1}^{M} D_{N}(z_m)}{\prod_{l=1}^{L} D_{N}(y_l) } \right]&=&
\frac{N!}{(N-L)!L!\mathcal{Z}_N} \left[ \prod_{n=1}^{N} \int_I dx_n \prod_{m=1}^M(z_m-x_n)\right] \det [\varphi_l(x_k)]_{k,l=1}^{N}
\nonumber\\
&&\times\frac{\prod_{k=1}^{L}\left(\frac{x_{k}}{y_k}\right)^{N-L}}{\prod_{n,l=1}^{L} (y_l - x_{n})}\
\Delta_L(x_{1},\ldots,x_{L})\Delta_{N-L}(x_{L+1},\ldots,x_{N})\ ,
\label{Estep1}
\eee
where we used the fact that each term in the sum over permutations gives the same contribution to the expectation.
Hence, we can undo the permutations under the sum by a change of variables, and replace the sum over $S_N/(S_{N-L} \times S_L)$ by the cardinality of the coset space $N!/(N-L)!L!$. Next, we expand the determinant over the $\det \left[ \varphi_l(x_k) \right]_{k,l=1}^{N} $ and then separate the integration over the first $L$ variables
$x_{l=1,\ldots,L}$ and the following $N-L$ variables $x_{n=L+1,\ldots,N}$, by also splitting the characteristic polynomials accordingly.
This gives
\begin{equation}\label{Estep2}
\begin{split}
&\mathbb{E} \left[\frac{\prod_{m=1}^{M} D_{N}(z_m)}{\prod_{l=1}^{L} D_{N}(y_l) } \right]\\
&=\frac{N!}{(N-L)!L!\mathcal{Z}_N} \sum_{\sigma \in S_N}\sgn(\sigma)
\left[ \prod_{l=1}^{L} \int_I dx_l \ \varphi_{\sigma(l)}(x_l)
\frac{x_{l}^{N-L}}{y_l^{N-L}}
\frac{\prod_{m=1}^M(z_m-x_l)}{\prod_{j=1}^{L} (y_j - x_{l})}\right]
\Delta_L(x_{1},\ldots,x_{L})\\
&\quad\times
\left[ \prod_{k=L+1}^{N} \int_I dx_k \ \varphi_{\sigma(k)}(x_k)
\prod_{m=1}^M(z_m-x_k)
\right]\Delta_{N-L}(x_{L+1},\ldots,x_{N})\ .
\end{split}
\end{equation}
Because we are aiming at an expression that will be amenable to taking the large-$N$ limit, we now focus on the integrals over $N-L$ variables in the second line, which we denote by $J$. Here, we  make use of one of the properties of the Vandermonde determinant, namely the absorption of the $M$ characteristic polynomials in $J$ into a larger Vandermonde determinant, see  \eqref{extensionVandermonde}, to write
\begin{equation}\label{Jdef}
J=\left[ \prod_{k=L+1}^{N} \int_I dx_k \ \varphi_{\sigma(k)}(x_k)\right]
\frac{\Delta_{N-L+M}(z_1,\ldots,z_M,x_{L+1},\ldots,x_{N})}{\Delta_{M}(z_{1},\ldots,z_{M})}.
\nonumber
\end{equation}
We use the 
representation \eqref{Vandermonde}, pull the integrations $\int_I dx_k \ \varphi_{\sigma(k)}(x_k)$ into the corresponding columns, and use definition \eqref{gentries} of the generalised moment matrix  to obtain
\begin{equation}
J=\frac{1}{\Delta_M(z_1,\ldots,z_M)} \left\vert
\begin{matrix}
z_1^{N+M-L-1} & \ldots & z_M^{N+M-L-1} &  g_{N+M-L,\sigma(L+1)} & \ldots & g_{N+M-L,\sigma(N)} \\
\vdots & \vdots & \vdots & \vdots & \vdots & \vdots  \\
z_1 & \ldots & z_M & g_{2,\sigma(L+1)} & \ldots & g_{2,\sigma(N)} \\
1 & \ldots & 1 & g_{1,\sigma(L+1)} & \ldots & g_{1,\sigma(N)} \\
\end{matrix}
 \right\vert.
 \nonumber
\end{equation}
Property (\ref{Gcondition}) of invertible polynomial ensembles enables us to rewrite $J$ as
\begin{equation}
\begin{split}
J &=\frac{1}{\Delta_M(z_1,\ldots,z_M)}\\
&\quad\times\left\vert
\begin{matrix}
z_1^{N+M-L-1} & \ldots & z_M^{N+M-L-1} &
\pi_{N+M-L-1}(a_{\sigma(L+1)}) & \ldots &
\pi_{N+M-L-1}(a_{\sigma(N)})
\\
\vdots & \vdots & \vdots & \vdots & \vdots & \vdots  \\
z_1 & \ldots & z_M & \pi_{1}(a_{\sigma(L+1)}) & \ldots & \pi_{1}(a_{\sigma(N)}) \\
1 & \ldots & 1 & \pi_{0}(a_{\sigma(L+1)})  & \ldots & \pi_{0}(a_{\sigma(N)}) \\
\end{matrix}
 \right\vert.\nonumber
 \end{split}
\end{equation}
Property (\ref{Picondition}) allows us to replace again the determinant of monic polynomials by a Vandermonde determinant of size $N-L+M$
to obtain
\begin{equation}
J=\frac{\Delta_{N-L}(a_{\sigma(L+1)},\ldots,a_{\sigma(N)})
}{\Delta_M(z_1,\ldots,z_M)}
\left[\prod_{j=1}^M\int_{I^\prime}dt_j F(t_j,z_j)
\prod_{n=L+1}^N(t_j-a_{\sigma(n)})
\right]
\Delta_M(t_1,\ldots,t_M)
\nonumber\\.
\end{equation}
Let us come back to the expectation value of characteristic polynomials in the form \eqref{Estep2} and insert what we have derived for $J$ above. This gives
\begin{equation}\label{Estep3}
\begin{split}
&\mathbb{E}\left[\frac{\prod_{m=1}^{M} D_{N}(z_m)}{\prod_{l=1}^{L} D_{N}(y_l) } \right]\\
&=\frac{N!}{(N-L)!L!\mathcal{Z}_N\Delta_M(z_1,\ldots,z_M)}
\left[\prod_{j=1}^M \int_{I^\prime}dt_j F(t_j,z_j)
\prod_{n=1}^N(t_j-a_{n})\right]
\Delta_M(t_1,\ldots,t_M)\\
&\quad\times\left[ \prod_{l=1}^{L} \int_I dx_l
\left(\frac{x_{l}}{y_l}\right)^{N-L}
\frac{\prod_{m=1}^M(z_m-x_l)}{\prod_{j=1}^{L} (y_j - x_{l})}\right]
\Delta_L(x_{1},\ldots,x_{L})\\
&\quad\times \sum_{\sigma \in S_N}\sgn(\sigma)\Delta_{N-L}(a_{\sigma(L+1)},\ldots,a_{\sigma(N)})
\prod_{l=1}^L
\frac{\varphi(a_{\sigma(l)},x_l)}{\prod_{j=1}^M(t_j-a_{\sigma(l)})}\ .
\end{split}
\end{equation}
The integrals are now put into a form to apply the following Lemma, that will allow us to simplify (and eventually get rid of) the sum over permutations.
\begin{lem}\label{LemmaSumSymmetricGroup}Let $S_N$ denote the permutation group of $\left\{1,\ldots,N\right\}$, and let $S_L$ be the subgroup of $S_N$ realized as the permutation group
of the first $L$ elements $\{1,\ldots,L\}$. Also, let $S_{N-L}$ be the subgroup of $S_N$ realised as the permutation group of the remaining $N-L$
elements $\{L+1,\ldots,N\}$.  Assume that $F$ is a complex valued function on $S_N$ which satisfies the condition
$F(\sigma h)=F(\sigma)$ for each $\sigma\in S_N$, and each $h\in S_L\times S_{N-L}$. Then we have
\begin{equation}\label{SumSymmetricGroup}
\sum\limits_{\sigma\in S_N}F(\sigma)=(N-L)!L!\sum\limits_{1\leq l_1<\ldots<l_L\leq N}F\left(\left(l_1,\ldots,l_{L},1,\ldots,\check{l}_1,\ldots,\check{l}_L\ldots,N\right)\right),
\end{equation} 
where $\left(i_1,\ldots,i_N\right)$ is a one-line notation for the permutation $\left(\begin{array}{cccc}
                                                                                      1 & 2 & \ldots & N \\
                                                                                      i_1 & i_2 & \ldots & i_N 
                                                                                    \end{array}
\right)$, and notation $\check{l}_p$ means that $l_p$ is removed from the list.
\end{lem}
\begin{proof} Recall that if $G$ is a finite group, and $H$ is its subgroup, then there are transversal elements 
$t_1,\ldots,t_k\in G$ for the left cosets of $H$ such that $G=t_1H\uplus\ldots\uplus t_kH$, where $\uplus$ 
denotes disjoint union. It follows that if $F$ is a function on $G$ with the property
$F(gh)=F(g)$ for any $g\in G$, and any $h\in H$, then
\begin{equation}\label{SumGroup}
\sum\limits_{g\in G}F(g)=|H|\sum\limits_{i=1}^kF\left(t_i\right),
\end{equation}
where $|H|$ denotes the number of elements in $H$. In our situation $G=S_N$, $H=S_L\times S_{N-L}$, and each transversal element 
can be represented as a permutation
$$
\left(l_1,\ldots,l_{L},1,\ldots,\check{l}_1,\ldots,\check{l}_L\ldots,N\right),
$$ 
written in one-line notation,
where $1\leq l_1<\ldots<l_L\leq N$. Moreover, each  collection of numbers $l_1$, $\ldots$, $l_L$ satisfying the condition $1\leq l_1<\ldots<l_L\leq N$  gives a transversal element 
for the left cosets of $H=S_L\times S_{N-L}$ in $G=S_N$. We conclude that equation (\ref{SumGroup}) is reduced to equation (\ref{SumSymmetricGroup}). 
\end{proof}
Assume that  $\Phi(x_1,\ldots,x_L)$ is  antisymmetric  under permutations $\sigma$ of its $L$ variable, i.e.
$$
\Phi(x_{\sigma(1)},\ldots,x_{\sigma(L)})=\sgn(\sigma)\Phi(x_1,\ldots,x_L),
$$
and that $L\leq N$. Let $F$ be the function on $S_N$ defined by
\begin{equation}
\begin{split}
&F(\sigma)=\sgn(\sigma)  \Delta_{N-L}(a_{\sigma(L+1)},\ldots,a_{\sigma(N)})
\left[ \prod_{k=1}^L\int_Idx_k f(a_{\sigma(k)},x_k)\right]\Phi(x_1,\ldots,x_L)\ ,
\\
\end{split}
\nonumber
\end{equation}
where $f$ is a function of two variables. Clearly, $F$ satisfies the condition $F(\sigma h)=F(\sigma)$ for each $\sigma\in S_N$, and each $h\in S_L\times S_{N-L}$.
Application of Lemma \ref{LemmaSumSymmetricGroup} to this function gives
\begin{equation}
\begin{split}
\sum\limits_{\sigma\in S_N}F(\sigma)&=(N-L)!L!
\sum\limits_{1\leq l_1<\ldots<l_L\leq N}\sgn\left(\left(l_1,\ldots,l_{L},1,\ldots,\check{l}_1,\ldots,\check{l}_L\ldots,N\right)\right)
\\
&\quad\times\Delta_{N-L}^{(l_1,\ldots,l_L)}\left(a_1,\ldots,a_N\right)
\left[ \prod_{k=1}^L\int_Idx_k f(a_{l_k},x_k)\right]\Phi(x_1,\ldots,x_L)\ ,
\end{split}
\nonumber
\end{equation}
where the reduced Vandermonde determinant is defined in \eqref{reducedVandermonde}. 
Taking into account that 
\begin{equation}
\sgn\left(\left(l_1,\ldots,l_{L},1,\ldots,\check{l}_1,\ldots,\check{l}_L\ldots,N\right)\right)=(-1)^{l_1+\ldots+l_L-\frac{L(L+1)}{2}},
\end{equation}
we obtain the formula
\begin{equation}\label{IntegrationFormula}
\begin{split}
&\sum\limits_{\sigma\in S_N}\sgn(\sigma)\Delta_{N-L}(a_{\sigma(L+1)},\ldots,a_{\sigma(N)})\left[ \prod_{k=1}^L\int_Idx_k f(a_{\sigma(k)},x_k)\right]\Phi(x_1,\ldots,x_L)\\
&=(N-L)!L!
\sum\limits_{1\leq l_1<\ldots<l_L\leq N}(-1)^{l_1+\ldots+l_L-\frac{L(L+1)}{2}}
\Delta_{N-L}^{(l_1,\ldots,l_L)}\left(a_1,\ldots,a_N\right)\\
&\quad \times\left[ \prod_{k=1}^L\int_Idx_k f(a_{l_k},x_k)\right]\Phi(x_1,\ldots,x_L)\ ,
\end{split}
\end{equation}
valid for any antisymmetric function $\Phi(x_{\sigma(1)},\ldots,x_{\sigma(L)})$, and for any function
$f(x,y)$ such that the integrals in the equation above exist.

Formula (\ref{IntegrationFormula}) enables us to rewrite equation (\ref{Estep3}) as
\begin{equation}\label{af}
\begin{split}
&\mathbb{E}\left[\frac{\prod_{m=1}^{M} D_{N}(z_m)}{\prod_{l=1}^{L} D_{N}(y_l) } \right]\\
&=\frac{N!}{\mathcal{Z}_N\Delta_M(z_1,\ldots,z_M)}
\left[\prod_{j=1}^M \int_{I^\prime}dt_j F(t_j,z_j)
\prod_{n=1}^N(t_j-a_{n})\right]
\Delta_M(t_1,\ldots,t_M)\\
&\quad\times\left[ \prod_{l=1}^{L} \int_I dx_l
\left(\frac{x_{l}}{y_l}\right)^{N-L}
\frac{\prod_{m=1}^M(z_m-x_l)}{\prod_{j=1}^{L} (y_j - x_{l})}\right]
\Delta_L(x_{1},\ldots,x_{L})
\\
&\quad\times
\sum\limits_{1\leq l_1<\ldots<l_L\leq N}(-1)^{l_1+\ldots+l_L-\frac{L(L+1)}{2}}
\Delta_{N-L}^{(l_1,\ldots,l_L)}\left(a_1,\ldots,a_N\right)\prod_{i=1}^L
\frac{\varphi(a_{l_i},x_i)}{\prod_{j=1}^M(t_j-a_{l_j})}.
\end{split}
\end{equation}
We note that due to \eqref{reducedVandermondeGeneral} it holds
\begin{equation}
\frac{\Delta_{N-L}^{(l_1,\ldots,l_L)}\left(a_1,\ldots,a_N\right)}{\Delta_{N}\left(a_1,\ldots,a_N\right)}
=\frac{(-1)^{l_1+\ldots+l_L-L} \Delta_{L}\left(a_{l_1},\ldots,a_{l_L}\right) }{\underset{n\neq l_1}{\prod\limits_{n=1}^N}\left(a_{l_1}-a_n\right)\ldots\underset{n\neq l_L}{\prod\limits_{n=1}^N}\left(a_{l_L}-a_n\right)}.
\end{equation}
In addition, we   apply \eqref{ZNinv} to eliminate $\mathcal{Z}_N$,  cancel  signs, and see that the strict ordering of the indices $l_1<l_2<\ldots<l_L$ can be relaxed,
\be
L! \sum_{1\leq l_1<\ldots<l_L \leq N} \rightarrow \sum_{l_1=1}^N\cdots\sum_{l_L=1}^N\ .
\nonumber
\ee
Finally, we see that the sum in formula (\ref{af}) can be written as contour integrals, because of the formula
\be\label{Residue}
\frac{1}{2\pi i} \oint_{C} du \frac{f(u)}{\prod_{n=1}^{N} (u-a_n)} = \sum_{l=1}^{N} \frac{f(a_l)}{\prod_{\substack{n=1 \\ n \neq l}}^{N} (a_l -a_n)}\ ,
\ee
where the contour $C$ encircles the points $a_1,\ldots,a_N$ counter-clockwise.
The leads to the formula in the statement of Theorem \ref{MainTheorem}.
\qed


\section{Special cases}\label{SectionSpecialCases}
In Proposition \ref{PropositionKernel} we have used  equation (\ref{vevgen}) in the case $M=L=1$. Another case of interest is that corresponding to products of characteristic polynomials. In this case $L=0$, and we obtain that only 
the first set of integrals remains in \eqref{vevgen}, i.e.
\be
\mathbb{E}\left[{\prod_{m=1}^{M} D_{N}(z_m)}
\right]=
\frac{\det[\mathcal{B}_{i}(z_{j})]_{i,j=1}^M}{\Delta_M(z_1,\ldots,z_M)}\ ,
\ee
where
\be
\label{Bdef}
\mathcal{B}_{i}(z) = \int_{I^\prime}ds F(s,z)\,s^{M-i}
\prod_{n=1}^N(s-a_{n})\ ,
\ee
after pulling the $M$ integrations over the $s_j$'s into the Vandermonde determinant of size $M$.
This result also could have been directly computed  using Lemma \ref{extensionVandermonde}.

As a final special case of interest we look at the ratio of $M+1$ characteristic polynomials over a single one at $L=1$. This object is needed in the application to finite temperature QCD, cf. \cite{SWG}. Theorem \ref{MainTheorem} gives
\bee
&&\mathbb{E}\left[\frac{\prod_{m=1}^{M+1} D_{N}(z_m)}{D_{N}(y) } \right]\nonumber\\
&=&\frac{1}{\Delta_{M+1}(z_1,\ldots,z_{M+1})}
\left(\prod_{j=1}^{M+1} \int_{I^\prime}ds_j F(s_j,z_j)
\prod_{n=1}^N(s_j-a_{n})\right)\Delta_{M+1}(s_1,\ldots,s_{M+1})
\nonumber\\
&&\times
\int_I dv \left(\frac{v}{y}\right)^{N-1}
\frac{\prod_{m=1}^{M+1}(z_m-v)}{(y - v)}
\oint_{C} \frac{du}{2\pi i} \frac{1}{\prod_{n=1}^{N} (u-a_n)}
\frac{\varphi(u,v)}{\prod_{j=1}^{M+1}(s_j-u)}\ .
\label{EM1}
\eee
Following \cite{DesFor}, we may use the Lagrange extrapolation formula
\be
\frac{1}{\prod_{j=1}^{M+1}(u-s_j)}=\sum_{m=1}^{M+1}\frac{1}{u-s_m}\prod_{\substack{j=1\\j\neq m}}^{M+1}\frac{1}{s_m-s_j}\ ,
\ee
to rewrite
\be
\frac{1}{\prod_{j=1}^{M+1}(s_j-u)}\Delta_{M+1}(s_1,\ldots,s_{M+1})
=(-1)^{M+1}\sum_{m=1}^{M+1}\frac{(-1)^{m-1}}{u-s_m}\Delta_M^{(m)}(s_1,\ldots,s_{M+1})\ .
\ee
This leads to the following rewriting of \eqref{EM1}
\bee
&&\mathbb{E} \left[\frac{\prod_{m=1}^{M+1} D_{N}(z_m)}{D_{N}(y) } \right]\nonumber\\
&=&\frac{ (-1)^{M} }{\Delta_{M+1}(z_1,\ldots,z_{M+1})}
\int_I dv \left(\frac{v}{y}\right)^{N-1}
\frac{\prod_{m=1}^{M+1}(z_m-v)}{(y - v)}
\oint_{C} \frac{du}{2\pi i} \frac{\varphi(u,v)}{\prod_{n=1}^{N} (u-a_n)}
\nonumber\\
&&\times
\sum_{m=1}^{M+1}(-1)^{m}
\left(\prod_{j=1}^{M+1} \int_{I^\prime}ds_j  F(s_j,z_j)
\prod_{n=1}^N(s_j-a_{n})\right)\frac{1}{u-s_m}\Delta_M^{(m)}(s_1,\ldots,s_{M+1})
\nonumber\\
&=& \frac{ (-1)^{M} }{\Delta_{M+1}(z_1,\ldots,z_{M+1})}
\int_I dv \left(\frac{v}{y}\right)^{N-1}
\frac{\prod_{m=1}^{M+1}(z_m-v)}{(y - v)}
\oint_{C} \frac{du}{2\pi i} \frac{\varphi(u,v)}{\prod_{n=1}^{N} (u-a_n)}
\nonumber\\
&&\times
\det \left[ \begin{matrix}
\mathcal{A}(z_1,u) & \ldots & \mathcal{A}(z_{M+1},u) \\
\mathcal{B}_{1}(z_1) & \ldots & \mathcal{B}_{1}(z_{M+1}) \\
\vdots & \ldots & \vdots \\
\mathcal{B}_{M}(z_1) & \ldots & \mathcal{B}_{M}(z_{M+1}) \\
\end{matrix} \right] ,
\label{EM1step2}
\eee
where we have defined
\be
\mathcal{A}(z,u) = \int_{I^\prime}ds F(s,z) \frac{-1}{u-s}
\prod_{n=1}^N(s-a_{n})\ .
\label{Adef}
\ee
In the second step in \eqref{EM1step2} we have first pulled all the $s$-integrals except the one over $s_m$ into the Vandermonde determinant $\Delta_M^{(m)}(s_1,\ldots,s_{M+1})$, leading to a determinant of size $M$ with matrix elements $\mathcal{B}_i(z_j)$ \eqref{Bdef}. We then recognise that the sum is a Laplace expansion of a determinant of size $M+1$ with respect to the first row, containing the matrix elements $\mathcal{A}(z_j,u)$ \eqref{Adef}.
This reveals the determinantal form of the corresponding kernel.

Fyodorov, Grela, and Strahov \cite{Fyodorov} considered the probability density defined by\footnote{We use a different convention than \cite{Fyodorov} where the factor $e^{-a_l}$ is part of the normalisation, cf. \eqref{ext2ev}.}
\begin{equation}
P_N^{\mathcal{L}}\left(x_1,\ldots,x_N\right)=\frac{1}{\mathcal{Z}_N^{\mathcal{L}}}\Delta_N\left(x_1,\ldots,x_N\right)
\det\left[x_k^{\mathcal{L}}e^{-(x_k+a_l)}I_0\left(2\sqrt{a_lx_k}\right)\right]_{k,l=1}^N 
\end{equation}
on $\mathbb{R}_+^N$. 
Note that this polynomial ensemble is invertible only for $\mathcal{L}=0$ as it follows from equations (\ref{GmatrixentriesLaguerre}) and (\ref{zkexample2}).
However, computations of different averages with respect to $P_N^{\mathcal{L}}$ can be reduced to those with respect
to $P_N^{\mathcal{L}=0}$, i.e. with respect  to an invertible ensemble. Indeed, we have
\begin{equation}\label{reduction}
\E_{P_N^{\mathcal{L}}}\left(f\left(x_1,\ldots,x_N\right)\right)
=\frac{\E_{P_N^{\mathcal{L}=0}}\left(f\left(x_1,\ldots,x_N\right)
\prod_{l=1}^{\mathcal{L}}D_N\left(z_l\right)\right)}{\E_{P_N^{\mathcal{L}=0}}\left(
\prod_{l=1}^{\mathcal{L}}D_N\left(z_l\right)\right)}\biggl|_{z_1,\ldots,z_{\mathcal{L}}=0}
\end{equation}
for any function $f\left(x_1,\ldots,x_N\right)$ such that the expectations in the formula above exist.
In particular,  we can reproduce the results of \cite{Fyodorov}  for the expectation value of a single characteristic polynomial, its inverse or a single ratio. Without going much into detail, we need two ingredients for this check. First, in order to perform the limit of vanishing arguments in equation (\ref{reduction}), it is useful to antisymmetrise the product of the first $\mathcal{L}$ functions $F(t_j,z_j)$ using the Vandermonde determinant $\Delta_{\mathcal{L}+1}(t_1,\ldots,t_{\mathcal{L}+1})$ in \eqref{vevgen}. We are then led to consider
\bee
\lim_{z_1,\ldots,z_{\mathcal{L}}\to0}\frac{\det\left[I_0(2\sqrt{z_it_j})\right]_{i,j=1}^{\mathcal{L}}}{\Delta_{\mathcal{L}}(z_1,\ldots,z_{\mathcal{L}})}&=&\lim_{z\to0}\det\left[\frac{t_i^{j-1}}{(j-1)!}\frac{I_{j-1}(2\sqrt{zt_i})}{\sqrt{zt_i}^{j-1}}\right]_{i,j=1}^{\mathcal{L}}
\nonumber\\
&=&\frac{(-1)^{\mathcal{L}(\mathcal{L}-1)/2}}{\prod_{j=1}^{\mathcal{L}}(j-1)!^2}\Delta_{\mathcal{L}}(t_1,\ldots,t_{\mathcal{L}})\ ,
\label{degenerate}
\eee
after first taking the limit of degenerate arguments, which is then sent to zero. Obviously, we first separate the remaining non-vanishing argument $z_{\mathcal{L}+1}$ from the Vandermonde determinant by
 $\Delta_{\mathcal{L}}(z_1,\ldots,z_{\mathcal{L}})\prod_{l=1}^{\mathcal{L}}(z_l-z_{\mathcal{L}+1})=\Delta_{\mathcal{L}+1}(z_1,\ldots,z_{\mathcal{L}+1})$.

Second, we need an equivalent formulation of Propositions 3.1 and 3.5 employed in \cite{Fyodorov}, which are due to \cite{ForresterLiu,DesFor}, respectively.\footnote{Notice the different convention for $\Delta_N$ used in \cite{Fyodorov,ForresterLiu,DesFor}.}
\begin{prop}\label{Prop3.1-5}
\bee
\mathbb{E}_{\mathcal{P}} \left[\frac{1}{D_{N}(y) } \right]
&=&\frac{1}{\det G}
\left\vert
\begin{matrix}
 g_{1,1} 
 & \ldots & g_{1,N} \\
 \vdots 
 & \ddots & \hdots  \\
g_{N-1,1} 
& \ldots & g_{N-1,N} \\
\int_0^\infty\frac{du\varphi_1(u)}{y-u}\left(\frac{u}{y}\right)^{N-1} 
&   \ldots &
\int_0^\infty\frac{du\varphi_N(u)}{y-u}\left(\frac{u}{y}\right)^{N-1} \\
\end{matrix}
 \right\vert
 \nonumber\\
 &=&
 \int_0^\infty\frac{du}{y-u}\left(\frac{u}{y}\right)^{N-1} \sum_{j=1}^N c_{N,j} \varphi_j(u)\ ,
 \label{Prop3factorN}
\eee
where $C$ is the inverse of the $N\times N$ moment matrix $G$, and $c_{i,j}$ are the matrix elements of $C^T$.
\end{prop}
\begin{proof}
Eqs. \eqref{Prop3factorN} were stated in \cite{Fyodorov} following \cite{ForresterLiu,DesFor}, without the factors of  $(u/y)^{N-1}$. The equivalence of the two statements can be seen as follows.
Expanding the geometric series inside the determinant without these factors we have
\bee
&&\frac{1}{\det G}
\left\vert
\begin{matrix}
 g_{1,1} 
 & \ldots & g_{1,N} \\
 \vdots 
 & \ddots & \vdots  \\
g_{N-1,1} 
& \ldots & g_{N-1,N} \\
\int_0^\infty{du\varphi_1(u)}\sum_{j=0}^\infty\frac{u^j}{y^{J+1}} 
&   \ldots &
\int_0^\infty{du\varphi_N(u)}\sum_{j=0}^\infty\frac{u^j}{y^{J+1}}
\\
\end{matrix}
 \right\vert
 \nonumber\\
 &&=
 \left\vert
\begin{matrix}
 g_{1,1}  & \ldots & g_{1,N} \\
 \vdots 
 & \ddots & \vdots  \\
g_{N-1,1} 
& \ldots & g_{N-1,N} \\
\sum_{j=N}^\infty\frac{g_{j,1}}{y^{J+1}} 
&   \ldots &
\sum_{j=N}^\infty\frac{g_{j,N}}{y^{J+1}}
\\
\end{matrix}
 \right\vert \det[c_{i,j}]_{i,j=1}^N\ .
 \nonumber
\eee
If we perform the integrals in the last row, we obtain infinite series over generalised moment matrices $g_{k,l}$, the first $N-1$ of which can be removed by subtraction of the upper $N-1$ rows. Rewriting the last row as integrals and resumming the series we arrive at the first line of \eqref{Prop3factorN}.

The second line in \eqref{Prop3factorN} is obtained as follows. Using that $\det[c_{i,j}]_{i,j=1}^N=1/\mathcal{Z}_N$ and then multiplying the matrix $C$ with the matrix inside the determinant from the right, this leads to an identity matrix, except for the last row, as $C$ is the inverse of the finite, $N\times N$ dimensional matrix. Laplace expanding with respect to the last column leads to the desired result.
\end{proof}

Employing Proposition \ref{Prop3.1-5} in \cite{Fyodorov}, it is not difficult to see that from our Theorem \ref{MainTheorem} together with \eqref{degenerate}
we obtain an equivalent form of
 \cite[Theorem 4.1]{Fyodorov} for a single characteristic polynomial,
\cite[Theorem 3.4]{Fyodorov}  for its inverse and
\cite[Theorem 5.1]{Fyodorov}  for a single ratio.

\section*{Acknowledgments}
G.A. and T.W. were 
supported by the German research council
DFG through the grant  CRC 1283 ``Taming uncertainty and profiting from randomness and low regularity in analysis, stochastics and their applications''. 
E.S. was partially supported by the BSF grant 2018248 ``Products of random matrices via the theory of symmetric functions''.

We thank Maurice Duits, Peter Forrester and Mario Kieburg for useful discussions, and  
the Department of Mathematics at the Royal Institute of Technology (KTH) Stockholm for hospitality (G.A. and T.W.).
The School of Mathematics and Statistics of the University of Melbourne is equally thanked for support and hospitality, where part of this work was completed (G.A.).
Last but not least we that thank one of the referees for many useful remarks. 

\appendix

\section{Properties of Vandermonde determinants}\label{Vandermondeappendix}
In this appendix we first define the Vandermonde determinant in various equivalent ways. We then collect several of its properties when extending the number of variables by multiplication, and when reducing it by division through the corresponding factors.

\begin{defn}\label{DeltaDef}
The Vandermonde determinant of $N$ pairwise distinct variables $x_1,\ldots,x_N$ is denoted by $\Delta_N (x_1,\ldots,x_N)$ and can be represented in the following equivalent ways:
\be\label{Vandermonde}
\begin{split}
\Delta_N (x_1,\ldots,x_N)&=
\det \left[x_j^{N-i} \right]_{i,j=1}^{N} = \prod_{1\leq i < j \leq N} (x_i-x_j)
= \left\vert
\begin{matrix}
x_1^{N-1} & \ldots & x_N^{N-1} \\
\vdots & \vdots & \vdots \\
x_1 & \ldots & x_N \\
1 & \ldots & 1  \\
\end{matrix}
\right\vert
\\
&=(-1)^{N(N-1)/2} \det \left[x_j^{i-1} \right]_{i,j=1}^{N}\ .
\end{split}
\ee
For $N=1$ it follows that $\Delta_1(x_1)=1$, and we also formally define $\Delta_0=1$ in the absence of parameters.
\end{defn}

The Vandermonde determinant can be extended from $N$ to $N+M$ variables in the following way, when multiplied by $M$  characteristic polynomials.
\begin{lem}\label{Vandermondelemma}
The following extension formula holds for a
Vandermonde determinant of size $N$.  Let the $M$ parameters $\{ z_1,\ldots,z_M \}$ be pairwise distinct. Then it holds that
\be\label{extensionVandermonde}
\begin{split}
&\prod_{m=1}^{M} \prod_{n=1}^{N} (x_n-z_m) \Delta_N(x_1,\ldots,x_N) = \frac{\Delta_{N+M}(x_1,\ldots,x_N,z_1,\ldots,z_M)}{\Delta_M(z_1,\ldots,z_M)}\ . \\
\end{split}
\ee
\end{lem}
\begin{proof}
We proceed by  induction over $M$.
Defining  $z_1 \equiv x_{N+1}$, the $M=1$ case can seen from inserting the definition  \eqref{Vandermonde} in product form
\be
\begin{split}
\prod_{n=1}^{N}(x_n-z_1) \Delta(x_1,\ldots, x_N) &= \prod_{n=1}^{N} (x_n-x_{N+1}) \prod_{1\leq i<j \leq N} (x_i-x_j)
= \prod_{1\leq i<j\leq N+1} (x_i-x_j) \\
&= \Delta_{N+1}(x_1,\ldots,x_N, z_1)\ .
\end{split}
\ee
We now assume that Eq. \eqref{extensionVandermonde} is valid for any $M$. The induction step $M \rightarrow M+1$ is straightforward:
\be
\begin{split}
\prod_{m=1}^{M+1}\! \prod_{n=1}^{N} (x_n-z_m) \Delta_N(x_1,\ldots,x_N)
&= \prod_{n=1}^{N} (x_n-z_{M+1})
\frac{\Delta_{N+M}(x_1,\ldots, x_N, z_1,\ldots, z_{M})}{\Delta_M(z_1,\ldots, z_{M}) }
\\
&=\frac{\Delta_{N+M+1}(x_1,\ldots, x_N, z_1,\ldots, z_{M+1})}{\Delta_{M+1}(z_1,\ldots, z_{M+1}) }\ .
\end{split}
\ee
Using the induction assumption, multiplying by a factor of unity 
$\frac{\prod_{l=1}^M(z_l-z_{M+1})}{\prod_{l=1}^M(z_l-z_{M+1})}$
and using the definition  \eqref{Vandermonde} in product form, the formula \eqref{extensionVandermonde} for $M+1$ follows.
\end{proof}
For extended Vandermonde determinants it holds that
\be\label{Vandermondeswitch}
\Delta_{N+M}(x_1,\ldots,x_N, z_1,\ldots,z_M) = (-1)^{NM} \Delta_{N+M}(z_1,\ldots,z_M,x_1,\ldots,x_N)\ ,
\ee
by permuting rows in the determinant form in \eqref{Vandermonde}.

Let us introduce a notation for the Vandermonde determinant with a reduced number of indices. For  $L\leq N$ ordered indices $l_1,\ldots,l_L$ we define the reduced Vandermonde determinant of size $N-L$ by
\be\label{reducedVandermonde}
\Delta^{(l_1,\ldots,l_L)}_{N-L}(x_1,\ldots,x_N) \equiv \Delta_{N-L}(x_1,\ldots,x_{l_1-1},x_{l_1+1},\ldots,x_{l_L-1},x_{l_L+1},\ldots,x_N)\ ,
\ee
where the parameters $x_j$ with $j=l_1,\ldots,l_L$ are absent.
From the Definition \ref{DeltaDef} we obtain that for $L=N$ both sides are equal to unity.
We obtain the reduced Vandermonde  by the following formula.
\begin{lem}\label{reducedVandermondeLemma}
For $N\geq L$ the Vandermonde determinant of size $N-L$ obtained by removing the variables $x_{l_j}$, with $1\leq l_1<\ldots<l_L\leq N$, from the variables $x_{1},\ldots,x_N$ can be obtained via
\be\label{reducedVandermondeGeneral}
\Delta^{(l_1,\ldots,l_L)}_{N-L}(x_1,\ldots,x_N)= \prod_{j=1}^{L} (-1)^{N-l_j}  \frac{\Delta_N(x_1,\ldots,x_N) \Delta_L(x_{l_1},\ldots,x_{l_L})}{\prod_{j=1}^{L} \prod_{\substack{n=1 \\ n \neq l_j}}^{N} (x_n - x_{l_j})}\ .
\ee
\end{lem}
\begin{proof}
The proof is again done by induction. For $L=1$ we have
for the right hand side of \eqref{reducedVandermondeGeneral}\footnote{In our conventions empty products equal to unity.}
\be
\begin{split}
 (-1)^{N-l_1}  \frac{\Delta_N(x_1,\ldots,x_N) }{ \prod_{\substack{n=1 \\ n \neq l_1}}^{N} (x_n - x_{l_1})}
&=\prod_{n=1}^{l_1-1}\frac{1}{(x_n-x_{l_1})}
\prod_{n=l_1+1}^{N}\frac{-1}{(x_n-x_{l_1})}\prod_{1\leq i < j \leq N} (x_i-x_j)
\\
&=
\prod_{\substack{1\leq i < j \leq N\\ i,j\neq l_1}} (x_i-x_j)
=\Delta^{(l_1)}_{N-1}(x_1,\ldots,x_N)\ .
\end{split}
\ee
For the induction step we assume that \eqref{reducedVandermondeGeneral} holds for any $N>L\geq 1$. From the definition of the reduced Vandermonde \eqref{reducedVandermonde} as a product it is not difficult to see that
\be
\begin{split}
\Delta^{(l_1,\ldots,l_L)}_{N-L}(x_1,\ldots,x_N)&=(-1)^{N-l_{L+1}}
\prod_{\substack{n=1\\n\neq l_1\ldots,l_{L+1}}}^N (x_n-x_{l_{L+1}})\
\Delta_{N-L-1}^{(l_1,\ldots,l_{L+1})}(x_1,\ldots,x_N)\\
&= \prod_{\substack{n=1\\n\neq l_{L+1}}}^N (x_n-x_{l_{L+1}})\
\frac{(-1)^{N-l_{L+1}}}{\prod_{j=1}^L(x_{l_j}-x_{l_{L+1}})}\ \Delta_{N-L-1}^{(l_1,\ldots,l_{L+1})}(x_1,\ldots,x_N)\ .
\end{split}
\ee
Using the induction assumption for the left hand side and
solving this equation for the reduced Vandermonde determinant of size $N-L-1$ on the right hand side, we obtain
\be
\begin{split}
\Delta_{N-(L+1)}^{(l_1,\ldots,l_{L+1})}(x_1,\ldots,x_N)&=\!
\frac{(-1)^{N-l_{L+1}}\! \prod_{j=1}^L(x_{l_j}-x_{l_{L+1}})}{\prod_{1=n\neq l_{L+1}}^N (x_n-x_{l_{L+1}})}
\frac{\Delta_N(x_1,\ldots,x_N) \Delta_L(x_{l_1},\ldots,x_{l_L})}{\prod_{j=1}^{L} (-1)^{N-l_j}  \prod_{\substack{n=1 \\ n \neq l_j}}^{N} (x_n - x_{l_j})}
\\
&= \prod_{j=1}^{L+1} (-1)^{N-l_j}  \frac{\Delta_N(x_1,\ldots,x_N) \Delta_{L+1}(x_{l_1},\ldots,x_{l_{L+1}})}{\prod_{j=1}^{L+1} \prod_{\substack{n=1 \\ n \neq l_j}}^{N} (x_n - x_{l_j})}\ ,
\end{split}
\ee
which finishes the proof.
\end{proof}



\end{document}